\documentclass[10pt]{article}
\usepackage{amsthm,amsmath,amssymb}
\usepackage{subfig,sidecap,caption}
\usepackage{paralist}
\usepackage[usenames,dvipsnames]{color}
\usepackage[pdftex,breaklinks,colorlinks,
    citecolor={BlueViolet}, linkcolor={Blue},urlcolor=Maroon]{hyperref}  
\usepackage{tikz,pgfkeys}
\usepackage{tkz-graph}
\usetikzlibrary{decorations.pathmorphing}
\usetikzlibrary{quotes}
\usetikzlibrary{arrows.meta}
\usetikzlibrary{arrows,shapes}
\usetikzlibrary{decorations.pathreplacing, decorations.shapes}
\usepackage{graphicx}
\usepackage{charter,eulervm}%
\usepackage{multirow,booktabs,array}
\usepackage{tabularx}
\usepackage[final,expansion=alltext,protrusion=true]{microtype}

\theoremstyle{plain} %remark,  
\newtheorem{theorem}{Theorem}[section]
\newtheorem{lemma}[theorem]{Lemma}

\newtheorem{proposition}[theorem]{Proposition}
\newtheorem{definition}{Definition}

\newcommand{\comment}[1]{\textbackslash\!\!\textbackslash {\em #1}}
\tikzstyle{vertex}  = [{fill=blue,circle,draw,inner sep=1pt}]

\title{A Naive Algorithm for Feedback Vertex Set}

\author{Yixin Cao\thanks{Department of Computing, Hong Kong
    Polytechnic University, Hong Kong,
    China. \href{mailto:yixin.cao@polyu.edu.hk} {\tt yixin.cao@polyu.edu.hk}}
 }

\date{}

\begin{document}
\maketitle

\begin{abstract}
Given a graph on $n$ vertices and an integer $k$, the feedback vertex set problem asks for the deletion of at most $k$ vertices to make the graph acyclic.  We show that a greedy branching algorithm, which always branches on an undecided vertex with the largest degree, runs in single-exponential time, i.e., $O(c^k\cdot n^2)$ for some constant $c$.  
 \end{abstract}

 \section{Introduction}
 All graphs in this paper are undirected and simple.  A graph $G$ is
 given by its vertex set $V(G)$ and edge set $E(G)$, whose
 cardinalities will be denoted by $n$ and $m$ respectively.  A set
 $V_-$ of vertices is a \emph{feedback vertex set} of graph $G$ if
 $G - V_-$ is acyclic, i.e., being a forest.  Given a graph $G$ and an
 integer $k$, the {feedback vertex set} problem asks whether $G$ has a
 {feedback vertex set} of at most $k$ vertices.

 The feedback vertex set problem was formulated from artificial
 intelligence, where a feedback vertex set is also called a {\em loop
   cutset}.  For each instance of the constraint satisfaction problem
 one can define a constraint graph, and it is well known that the
 problem can be solved in polynomial time when the constraint graph is
 a forest \cite{freuder-82-backtrack-free-search}.  Therefore, one way
 to solve the constraint satisfaction problem is to find first a
 minimum feedback vertex set of the constraint graph, enumerate all
 possible assignments on them, and then solve the remaining instance.
 Given an instance $I$ of the constraint satisfaction problem on $p$
 variables, and a feedback vertex set $V_-$ of the constraint graph,
 this approach can be implemented in $O(p^{|V_-|}\cdot |I|^{O(1)})$
 time \cite{dechter-87-fvs-csp}.  A similar application was found in
 Bayesian inference, also in the area of artificial intelligence
 \cite{pearl-88-probabilistic-reasoning}; more updated material can be
 found in the Ph.D. thesis of Bidyuk \cite{bidyuk-06-dissertation}.

 The feedback vertex set problem is NP-hard \cite{lewis-80-node-deletion-np}.  The
 aforementioned approach for solving the constraint satisfaction
 problem only makes sense when $|V_-|$ is fairly small.  This
 motivates the study of parameterized algorithms for the feedback vertex set
 problem, i.e., algorithms that find a feedback vertex set of size at most $k$ in
 time $f(k)\cdot n^{O(1)}$.
 Since earlier 1990s, a chain of parameterized algorithms have been
 reported in literature; for a complete list we refer to
 \cite{cao-15-ufvs}.  Instead of providing a new and improved
 algorithm, this paper considers a naive branching algorithm that {\em
   should have been discovered decades ago}.

 A trivial branching algorithm will work as follows.  It picks a
 vertex and branches on either including it in the solution $V_-$
 (i.e., deleting it from $G$), or marking it ``undeletable,'' until
 the remaining graph is already a forest.  This algorithm however
 takes $O(2^n)$ time.  A (rather informal) observation is that a
 vertex of a larger degree has a larger chance to be in a minimum
 feedback vertex set, thereby inspiring the following two-phase greedy
 algorithm for solving the feedback vertex set problem.
 If there are undecided vertices of degree larger than two after some
 preprocessing, then it always branches on an undecided vertex with
 the largest degree.  Believe it or not, this greedy algorithm,
 implemented in its most naive way, already beats most previous
 algorithms for this problem.

 \begin{theorem}\label{thm:main}
   The greedy algorithm can be implemented in $O(8^k\cdot n^2)$ time.
 \end{theorem}

 The use of the observations on degrees in solving the feedback vertex
 set problem is quite natural.
 Indeed, the research on parameterized algorithms and that on
 approximation algorithms for the feedback vertex set problem have
 undergone a similar process.  Early work used the cycle
 packing-covering duality, and hence ended with
 $O((\log{k})^{O(k)}\cdot n^{O(1)})$-time parameterized algorithms
 \cite{raman-02-fvs} and $O(\log{n})$-ratio approximation algorithms
 \cite{erdos-62-number-disjoint-circuits}, respectively, while the
 first $2$-approximation algorithm uses a similar greedy approach on
 high-degree vertices \cite{bafna-99-approximate-fvs}.  Indeed, all
 the four slightly different $2$-approximation algorithms for this
 problem are based on similar degree observations
 \cite{chudak-98-primal-dual-interpretation, fujito-96-vc-and-fvs}.
 So is the quadratic kernel of
 Thomass\'e~\cite{thomasse-10-kernel-fvs}.  There is also an
 $O(4^k\cdot n)$-time randomized algorithm
 \cite{becker-00-randomized-fvs} based on this idea.  Our greedy
 branching algorithm can be viewed as the de-randomization of this
 randomized algorithm.

 For a reader familiar with parameterized algorithms of the feedback
 vertex set problem, Theorem~\ref{thm:main} may sound somewhat
 surprising.  Deterministic single-exponential algorithms for the
 feedback vertex set problem had been sedulously sought, before
 finally discovered in 2005.  With so many different techniques, some
 very complicated, having been tried toward this end,\footnote{To
   date, the number of parameterized algorithms for feedback vertex
   set published in literature exceeds any other single problem,
   including the more famous vertex cover problem.} it is rather
 interesting that the goal can be achieved in such a naive way.

The significance of single-exponential algorithms for the feedback vertex set
problem lies also in the theoretical interest, for which let us put it
into context.

Together with the vertex cover problem (finding a set $V_-$ of at most
$k$ vertices of a graph $G$ such that $G - V_-$ is edgeless), the
feedback vertex set problem is arguably the most studied parameterized
problem.  However, a simple $O(2^k\cdot (m+n))$-time algorithm for
vertex cover was already known in 1980s \cite{mehlhorn-84-II}.  For
this difference there is a quick and easy explanation from the aspect
of graph modification problems \cite{lewis-80-node-deletion-np,
  cai-96-hereditary-graph-modification}.  Vertex deletion problems ask
for the deletion of a minimum set of vertices from a graph to make it
satisfy specific properties.  The vertex cover problem and the
feedback vertex set problem are precisely vertex deletion problems to,
respectively, the edgeless graphs and acyclic graphs, i.e., forests.
The obstruction (forbidden induced subgraph) for the edgeless graphs
is an single edge, the simplest one that is nontrivial.  On the other
hand, the obstructions for forests are all cycles, which may be
considered the simplest of all those infinite obstructions, for most
of which single-exponential algorithms are quite nontrivial, if
possible at all.

The problems vertex cover and feedback vertex set are also known as
planar-$\cal F$-deletion problems, which, given a graph $G$, a set
$\cal F$ of graphs of which at least one is planar, ask for a minimum
set of vertices whose deletion make the graph $H$-minor-free for every
$H\in\cal F$ \cite{fellows-88-nonconstructive-polynomial}.  They
correspond to the cases with ${\cal F} = \{K_2\}$ and
${\cal F} = \{K_3\}$ respectively.  Recently, Fomin et
al.~\cite{fomin-12-f-deletion} and Kim et
al.~\cite{kim-15-kernel-and-algorithms-via-protrusion-decompositions}
showed that all planar-$\cal F$-deletion problems can be solved in
single-exponential time.  With a huge constant hidden by the big-Oh,
their results, however, are of only theoretical interest.

Yet another way to connect the vertex cover problem and the feedback vertex set
problem is that a graph has treewidth zero if and only if it is
edgeless, and treewidth at most one if and only if it is a forest.
% Therefore, the feedback vertex set problem is also known as
The treewidth-two vertex deletion problem is planar-$\{K_4\}$-deletion
\cite{kim-15-k4-minor-cover}.

\section{The algorithm}

There is no secret in our algorithm, which is presented in Figure~\ref{fig:greedy-alg}, except the recursive form and an extra input $F$, the set of ``undeletable'' vertices.  We say that ($G, k, F$), where $F\subseteq V(G)$ and the solution is only picked from $V(G)\setminus F$, is an \emph{extended instance}; note that to make such an extended instance nontrivial, $F$ needs to induce a forest.
(Indeed, the solution in the original loop cutset problem has to be selected from ``allowed'' vertices, which is exactly the case $F$ comprising all vertices that are not allowed.)  The algorithm can be viewed as two parts, the first (steps 1--4) applying some simple operations when the situation is simple and clear, while the second (steps 5--7) trying both possibilities on whether a vertex $v$ is in a solution.  The operations in the first part are called reductions in the parlance of parameterized algorithms.  The three we use here are standard and well-known,\footnote{For the reader familiar with related algorithms, our reduction steps may seem slightly different from those in literature.  First of all, unlike most algorithms for the problem, our algorithm does not involve multiple edges.  We believe it is simple to keep the graph simple.  As a result, we are not able to eliminate all vertices of degree two: The common way to dispose of a vertex $v$ of degree-2 is to delete $v$ and add an edge between its two neighbors, so called \textit{smoothening}.  Smoothening a vertex whose two neighbors were already adjacent would introduce parallel edges.  Noting that there always exists an optimal solution avoiding $v$, one may move $v$ into $F$ \cite{cao-15-ufvs}, but we prefer the current form because it is simpler and easier to analyze.  It is also easier to be extended in Section~\ref{sec:improvement}.} and their correctness is straightforward; see, e.g., \cite{cao-15-ufvs}.

\begin{figure}[h!]
  \centering
  \tikz\path (0,0) node[draw, text width=.9\textwidth, rectangle, rounded corners, inner xsep=20pt, inner ysep=10pt]{
    \begin{minipage}[t!]{\textwidth}
\small {\bf Algorithm naive-fvs($G, k, F$)}
  \\
  {\sc Input}: a graph $G$, an integer $k$, and a set $F\subseteq V(G)$ inducing a forest.
  \\
  {\sc Output}: a feedback vertex set $V_-\subseteq V(G)\setminus F$
  of size $\leq k$ or ``\textsc{no}.''

  \begin{tabbing}
    AAA\=AAA\=aaa\=Aaa\=MMMMMMAAAAAAAAAAAAA\=A \kill
    0.\> {\bf if} $k<0$  {\bf then return ``\textsc{no}''}; {\bf if} $V(G)= \emptyset$  {\bf then return $\emptyset$};
    \\
    1.\> {\bf if} a vertex $v$ has degree less than two {\bf then}
    \\
    \>\> {\bf return naive-fvs}($G - \{v\}, k, F\setminus \{v\}$);
    \\
    2.\> {\bf if} a vertex $v\in V(G)\setminus F$ has two neighbors in the same component of $G[F]$ {\bf then}
    \\
    \>\> $X\leftarrow$ {\bf naive-fvs}($G - \{v\}, k - 1, F$);
    \\
    \>\>{\bf return} $X\cup\, \{v\}$;
    \\
    3.\> pick a vertex $v$ from $V(G)\setminus F$ with the maximum degree;
    \\
    4.\> {\bf if} $d(v) = 2$ {\bf then}
    \\
    4.1.\>\> $X\leftarrow \emptyset$;
    \\
    4.2.\>\> {\bf while} there is a cycle $C$ in $G$ {\bf then} 
    \\
    \>\>\> take any vertex $x$ in $C\setminus F$;
    \\
    \>\>\> add $x$ to $X$ and delete it from $G$;
    \\
    4.3.\>\> {\bf if} $|X| \le k$ {\bf then return} $X$; {\bf else return} ``\textsc{no}'';
    \\
    5.\> $X\leftarrow$ {\bf naive-fvs}($G - \{v\}, k - 1, F$);
    \comment{case 1: $v\in V_-$.}
    \\
    \> {\bf if} $X$ is not ``\textsc{no}''  {\bf then return}  $X\cup\, \{v\}$;
    \\
    6.\> {\bf return naive-fvs}($G, k, F\cup \{v\}$). \comment{case 2: $v\not\in V_-$.}
\end{tabbing}
    \end{minipage}
  };
\caption{A simple algorithm for feedback vertex set branching in a greedy manner.}
\label{fig:greedy-alg}
\end{figure}

\begin{lemma}
  Calling algorithm \texttt{naive-fvs} with ($G, k, \emptyset$) solves the instance ($G, k$) of the feedback vertex set problem.
\end{lemma}
\begin{proof}
  The two termination conditions in step~0 are clearly correct.
  For each recursive call in steps~1 and 2, we show that the original instance is a yes-instance if and only if the new instance is a yes-instance.  
  Note that no vertex is moved to $F$ in these two steps.
  In step~1, the vertex $v$ is not in any cycle, and hence it can be avoided by any solution.  In step~2, there is a cycle consisting of the vertex $v$ and vertices in $F$ (any path connecting these two vertices in $G[F]$), and hence any solution has to contain $v$.

  To argue the correctness of step~4, we show that the solution found in step~4 is optimal.  Let $c$ be the number of components in $G$ and $\ell$ the size of optimal solutions.  Note that every vertex in a solution has degree two, and hence after deleting $\ell$ vertices the graph has $n - \ell$ vertices and at least $m - 2\ell$ edges.  Moreover, deleting vertices from an optimal solution will not decrease the number of components of the graph, we have $(n - \ell) - (m - 2\ell) \ge c$.  Hence, $\ell \ge m - n + c$, and showing $|X| = m - n + c$ would finish the task.
Deleting a vertex of degree $2$ from a cycle never increases the number of components.  Also note that a graph on $c$ components contains a cycle if and only if it has more than $n - c$ edges.  Therefore, the while loop in step~4 would be run exactly $m - n + c$ iterations: After deleting $m - n + c$ vertices, each of degree two when deleted, the remaining graph has $2n - m - c$ vertices and $2n - m - 2 c$ edges, which has to be a forest of $c$ trees.
  
  The last two steps are trivial: If there is a solution containing $v$, then it is found in step~5; otherwise, step~6 always gives the correct answer.
\end{proof}

We now analyze the running time of the algorithm, which is simple but nontrivial.
The execution of the algorithm can be described as a search tree in which each node corresponds to two extended instances of the problem, the \emph{entry instance} and the \emph{exit instance}.   The entry instance of the root node is ($G, k, \emptyset$).   The exit instance of a node is the one after steps~1--3 have been exhaustively applied on the entry instance.  If step~5 is further called, then two children nodes are generated, with entry instances ($G - \{v\}, k - 1, F$) and ($G, k, F\cup \{v\}$) respectively.  (Note that the second child may not be explored by the algorithm, but this is not of our concern.)   A leaf node of the search tree returns either a solution or ``\textsc{no}.''

It is clear that each node {can be processed in polynomial time}, and thus the focus of our analysis is to bound the number of nodes in the search tree.  Since the tree is binary, it suffices to bound its depth. We say that a path from the root of the search tree to a leaf node is an \emph{execution path}.
% We bound the length of each execution path.
Let us fix an arbitrary execution path in the search tree of which the leaf node returns a solution $V_-$, and let $F'$ denote all the vertices moved into $F$ by step~6 in this execution path.  The length of this execution path is at most $|V_-| + |F'|$: Each non-leaf node puts at least one vertex to $V_-$ or $F$.  We are allowed to put at most $k$ vertices into $V_-$, i.e., $|V_-|\le k$, and hence our task in the rest of this section is to bound $|F'|$.

Let us start from some elementary facts on trees.  Any tree $T$ satisfies
\[
\sum_{v\in V(T)} d(v) = 2 |E(T)| = 2 |V(T)| - 2 \qquad \text{and}
\qquad \sum_{v\in V(T)} ( d(v) - 2 ) = - 2.
\]
Let $L$ denote the set of leaves of $T$, and $V_3$ the set of vertices of degree at least three.  If $V(T) \ge 2$, then $|L|\ge 2$ and 
\[
  -2 = \sum_{v\in L} ( d(v) - 2 ) + \sum_{v\in V(T)\setminus L} ( d(v) - 2 )
   = \sum_{v\in L} ( -1 ) + \sum_{v\in V_3} ( d(v) - 2 ) = \sum_{v\in V_3} ( d(v) - 2 ) - |L|.
\]
Hence
\begin{equation}   \label{eq:main}
  \sum_{v\in V_3} (d(v) - 2) = |L| - 2.
\end{equation}

The implication of \eqref{eq:main} for our problem is that the more
large-degree vertices ($V_3$) in the final forest $G - V_-$, the more
leaves ($L$) it has.  Every vertex $u\in F'$ will be in the forest.  Since its original degree is at least three, either its degree is decreased to two or less, or there must be some leaves produced to ``balance the equation \eqref{eq:main}.''  On the other hand, however, every vertex has degree at least two when $u$ is moved to $F$.  Therefore, if it is the second case, the leaves have to be ``produced'' in later steps.  The requirement of degree decrements is decided by the degree of $u$, and can be satisfied by vertices deleted later, whose degrees cannot be larger than that of $u$.
This informal observation would enable us to derive the desired lower bound on $|F'|$.

The following invariants will be used in our formal analysis.
\begin{description}
\item [Invariant 1]: During the algorithm, the degree of no vertex can increase.
\item [Invariant 2]: When a recursive call is made in step 5~or 6, there is no vertex of degree 0 or 1 in the graph.
\end{description}

This algorithm never directly deletes any edge, and thus the degree of a vertex decreases only when some of its neighbors are deleted from the graph,---we are talking about the degree in the whole graph $G$, so moving a vertex to $F$ does not change the degree of any vertex.  In particular, only steps 1, 2, 4, and 5 can decrease the degree of vertices.  
By Invariant~2, after a vertex is moved to $F$, step~1 cannot be called before step~2, 4, or 5.  In other words, the degree of a vertex in $F$ decreases only after some vertex put into $V_-$.  We can \emph{attribute} them to vertices $V_-$ as follows.

For a vertex $v\in V_-\cup F'$, we use $d^*(v)$ to denote the degree of $v$ at the moment it is deleted from the graph and put into $V_-$ (step 2, 4, or 5) or moved into $F$ (step 6).  Note that $d_{G}(v)\ge d^*(v)$ by Invariant~1, and $d^*(v)\ge 3$ when $v\in F'$.
Let $x_1$, $x_2$, $\ldots$, $x_{|V_-|}$ be the vertices in $V_-$, in the order of them being put into $V_-$, and let ($G_i, k_i, F_i$) be the exit instance in the node of the search tree corresponding to $x_i$.

\begin{definition}
We say that the decrements of the degree of a vertex $u \in F'$ from $d^*(u)$ to $2$ are \emph{effective}, and an effective decrement is \emph{incurred by $x_i\in V_-$} if it happens between deleting $x_i$ and $x_{i+1}$, or after deleting $x_{|V_-|}$ if $i = {|V_-|}$.  Let $\delta(u, x_i)$ denote the number of effective decrements of $u$ incurred by $x_i$.
\end{definition}
{Note that $\delta(u, x_i)$ may be larger than $1$.}
  It is worth stressing that we do not count the degree decrements of $u$ before it is moved into $F$.  Therefore, $\delta(u, x_i)$ can be positive only when $u$ is in $F$ when $x_i$ is deleted, i.e., $u\in F_i$ and hence $d_{G_i}(u)\ge 2$:
\[
  \delta(u, x_i) =
  \begin{cases}
    d_{G_i}(u) - \max\{ d_{G_{i + 1}}(u), 2\} &\text{ when }u\in F_i,
  \\
  0 &\text{ otherwise.} 
  \end{cases}
\]
\begin{proposition}\label{lem:degree-relation}
  For any $u\in F'$ and $x_i\in V_-$, if $\delta(u, x_i)> 0$ then $d^*(u) \ge d^*(x_i)$.
\end{proposition}

First, we bound the total number of effective decrements incurred by
$x_i$ for each $x_i\in V_-$.
\begin{lemma}\label{lem:effective-decrements}
  For each $x_i\in V_-$, it holds $\sum_{u\in F'} \delta(u, x_i)\le d^*(x_i)$.
\end{lemma}
\begin{proof}
  Recall that all effective decrements incurred by $x_i$ happen after deleting $x_i$ from $G_i$.  If $d_{G_i}(v) > 2$ for every vertex $v\in N_{G_i}(x_i)$, then the deletion of $x_i$ will not make the degree of any vertex smaller than two.  Therefore, step~1 will not be called before putting the next vertex into $V_-$.  The degree of each vertex in $N_{G_i}(x_i)$ decreases by one, and the total number of effective decrements incurred by $x_i$ is thus at most $d^*(x_i)$.

  In the rest $d_{G_i}(v) = 2$ for some $v\in N_{G_i}(x_i)$, and it becomes $1$ with the deletion of $x_i$.  This decrement is not effective, but it will trigger step~1, which may subsequently lead to effective decrements.
  Let $d$ denote the number of degree-2 neighbors of $x_i$ in $G_i$.  After the deletion of $x_i$, all of them have degree one, and there is no other vertex having degree one in $G_i - \{x_i\}$ (Invariant 2).  We consider the application of step~1, and let $x$ be the vertex deleted.  If the only neighbor of $x$ has degree two when this step is executed, then its degree becomes $1$ after the deletion of $x$, and hence the number of degree-1 vertices is not changed.  Otherwise, there is one less vertex of degree 1 but there may be  one effective decrement (only when the only neighbor of $x$ is in $F$ and has degree at least three).  Therefore, when step~1 is no longer applicable, the total number of effective decrements is at most $d^*(x_i) - d + d = d^*(x_i)$.
\end{proof}

We are now ready to bound the number of calls of step~6 made in this execution path, i.e.,  the size of $F'$, by the size of $V_-$.  This is exactly the place the greedy order of branching plays the magic.  
\begin{lemma}\label{thm:bound}
  In an execution path that leads to a solution, $|F'| \le 3|V_-|$.
\end{lemma}
\begin{proof}
  Since this execution path leads to a solution, all vertices must be deleted from the graph at the end of the path.   In the algorithm, a vertex in $F$ can only be deleted from the graph in step~1, when the degree of the vertex has to be one or zero.  On the other hand, $d^*(u) \ge 3$.  Thus, all the $d^*(u) - 2$ effective decrements must have happened on this vertex, i.e., $\sum_{v\in V_-}\delta(u, v) = d^*(u) - 2$.  Putting everything together, we have
  \begin{align*}
    |V_-| = \sum_{v\in V_-} 1
    &= \sum_{v\in V_-} \frac{d^*(v)}{d^*(v)}
    \\
    &\ge \sum_{v\in V_-} \frac{ 1 }{d^*(v)}
      \sum_{u\in F'} { \delta(u, v) }
      \tag{Lemma~\ref{lem:effective-decrements}}
    \\
    &= \sum_{v\in V_-} 
      \sum_{u\in F'} \frac{ \delta(u, v) }{d^*(v)}
    \\
    &\ge \sum_{v\in
      V_-} \sum_{u\in F'} \frac{ \delta(u, v) }{d^*(u)}
      \tag{Proposition~\ref{lem:degree-relation}}
    \\
    &=\sum_{u\in F'} \frac{ 1 }{d^*(u)} \sum_{v\in V_-} { \delta(u, v) }
    \\
    &= \sum_{u\in F'} \frac{ d^*(u) - 2 }{d^*(u)}
    \\
    &\ge \sum_{u\in F'} \frac{ 1 }{3}
    \tag{$d^*(u)\ge 3$}
    \\
    &= \frac{ |F'| }{ 3 },
  \end{align*}
  and the proof is complete.
\end{proof}

\begin{theorem}
  Algorithm \texttt{naive-fvs} can be implemented in
  $O(16^k\cdot n^2)$ time to decide whether a graph $G$ has a feedback vertex set
  of size at most $k$.
\end{theorem}
\begin{proof}
  If the input graph $G$ has a feedback vertex set of size at most
  $k$, then there must be an execution path that returns a solution,
  and by Lemma~\ref{thm:bound}, the length of this path is at most
  $4 k$.
  Otherwise, all execution paths return ``\textsc{no},'' disregard of their lengths.
  Therefore, we can terminate every execution path after it has put
  $3 k$ vertices into $F$ by returning ``\textsc{no}'' directly.  The new
  search tree would then have depth at most $4 k$.  Clearly, the
  processing in each node can be done in $O(n^2)$ time.  This gives
  the running time $O(2^{4 k + 1}\cdot n^2) = O(16^k\cdot n^2)$.
\end{proof}

\section{An improved running time}\label{sec:improvement}

It is long (but \emph{not} well) known that if the maximum degree of a
graph is at most three, then a minimum feedback vertex set can be found in
polynomial time \cite{furst88, ueno-88-cubic-fvs}.  This can be
extended to the setting that the degree bound holds only for the
undecided vertices i.e., vertices in $V(G)\setminus F$.
\begin{lemma}[\cite{cao-15-ufvs}]\label{thm:subcubic}
  Given a graph $G$ and a set $F$ of vertices such that every vertex
  in $V(G)\setminus F$ has degree at most three, there is a
  polynomial-time algorithm for finding a minimum set $V_-\subseteq
  V(G)\setminus F$ such that $G - V_-$ is a forest.
\end{lemma}

Therefore, we can change step~4 of algorithm \texttt{naive-fvs} to the following:

\begin{figure}[h]
  \centering\small
\begin{tabbing}
    AAA\=AAA\=aaa\=Aaa\=MMMMMMAAAAAAAAAAAAA\=A \kill
    4.\> {\bf if} $d(v)\le 3$ {\bf then}
    \\
    4.1.\>\> call Lemma~\ref{thm:subcubic} to find a minimum solution $X$;
    \\
    4.2.\>\> {\bf if} $|X| \le k$ {\bf then return} $X$; {\bf else return} ``\textsc{no}'';
\end{tabbing}
\end{figure}
\noindent Therefore, $d^*(u) \ge 4$ for each vertex $u\in F'$.   As a result, in the last inequality in the proof of
Lemma~\ref{thm:bound}, we can use
$({ d^*(u) - 2 }) / d^*(u) \ge 2 / 4 = 1/ 2$, which implies
$|F'| \le 2 |V_-|$.  The algorithm would then run in $O(8^k\cdot n^{O(1)})$
time.

We conclude this paper by pointing out that the analysis is not tight.  The inequalities in the proof of Lemma~\ref{thm:bound} can be tight only when $d^*(v) = 4$ for all vertices $v\in V_-\cup F$, and more importantly, all the degree decrements incurred by putting a vertex to $V_-$ are effective.  If such a graph exists,---we may assume without loss of generality that it does not contains any vertex of degree two or less,---then all its vertices have degree four, and all neighbors of a vertex $x\in V_-$ are in $F$.  But in such a graph there should be a different solution, and note that our algorithm only explore the subtree rooted at the child node made by step~6 \textit{only if} all the leaves in the other subtree (rooted at the node made by step~5) return ``\textsc{no}.''
% In other words, the worst case is a $4$-regular graphs with $3k - 1$ vertices.  

% \textcolor{blue}{For the reader familiar with parameterized complexity, we would like to point out our algorithm does not employ the technique so-called iterative compression.}

\paragraph{Acknowledgment.}  The author would like to thank O-joung
Kwon and Saket Saurabh for pointing out a mistake in the introduction
of the previous version.

\end{document}